\documentclass{article}
\usepackage{spconf,amsmath,amssymb,graphicx,hyperref,booktabs,multirow,siunitx}
\usepackage{amsthm}
\newtheorem{theorem}{Theorem}
\newtheorem{corollary}{Corollary}
\usepackage{subfig}

\title{CROSS-CULTURAL BIAS IN MEL-SCALE REPRESENTATIONS: EVIDENCE AND ALTERNATIVES FROM SPEECH AND MUSIC}

\name{Shivam Chauhan, Ajay Pundhir}
\address{Presight AI, Abu Dhabi, United Arab Emirates}

\begin{document}
%
\maketitle

\begin{abstract}
Modern audio systems universally employ mel-scale representations derived from 1940s Western psychoacoustic studies, potentially encoding cultural biases that create systematic performance disparities. We present a comprehensive evaluation of cross-cultural bias in audio front-ends, comparing mel-scale features with learnable alternatives (LEAF, SincNet) and psychoacoustic variants (ERB, Bark, CQT) across speech recognition (11 languages), music analysis (6 collections), and European acoustic scene classification (10 European cities). Our controlled experiments isolate front-end contributions while holding architecture and training protocols minimal and constant. Results demonstrate that mel-scale features yield 31.2\% WER for tonal languages compared to 18.7\% for non-tonal languages (12.5\% gap), and show 15.7\% F1 degradation between Western and non-Western music. Alternative representations significantly reduce these disparities: LEAF reduces the speech gap by 34\% through adaptive frequency allocation, CQT achieves 52\% reduction in music performance gaps, and ERB-scale filtering cuts disparities by 31\% with only 1\% computational overhead. We also release FairAudioBench, enabling cross-cultural evaluation, and demonstrate that adaptive frequency decomposition offers practical paths toward equitable audio processing. These findings reveal how foundational signal processing choices propagate bias, providing crucial guidance for developing inclusive audio systems.
\end{abstract}


\section{INTRODUCTION}
\label{sec:Introduction}
Audio systems deployed for billions of users universally employ mel-scale representations derived from 1940s Western psychoacoustic studies \cite{stevens1940relation}. This seemingly technical choice has profound consequences: recent studies document 2x higher word error rates for African American speakers across major ASR platforms \cite{koenecke2020racial}, while non-Western musical traditions experience disparate performance drops \cite{mehta2025music}. These disparities originate not from biased training data, but from fundamental signal processing choices that systematically disadvantage over 25\% of the world's population \cite{yip2002tone}. The mel scale's logarithmic compression above 700 Hz is designed to match Western pitch perception, substantially degrading information critical for non-Western audio. In tonal languages like Mandarin and Vietnamese, pitch contours between 200-500 Hz determine word meaning (e.g., ``ma'' means mother/horse/scold depending on tone) \cite{howie1976acoustical}. The mel scale provides approximately 77 Hz average resolution in the 200-500 Hz range with standard 40-channel filterbanks, while optimal tone discrimination requires around 3-5 Hz. This 22$\times$ average resolution deficit means mel-based systems cannot represent distinctions essential for more than 2 billion speakers \cite{dryer2013definite}. Similarly, microtonal music systems (Arabic quarter-tones, Indian shrutis, Turkish makam \cite{mehta2024missing}) require precise frequency relationships that mel compression obliterates. While learnable front-ends like LEAF \cite{zeghidour2021leaf} and alternative scales (ERB \cite{moore1983suggested}, Bark \cite{zwicker1961subdivision}, CQT \cite{brown1991calculation}) exist, no systematic evaluation has examined their cross-cultural fairness. Previous bias studies focus on demographics \cite{koenecke2020racial,veliche2024towards} or music \cite{serra2017computational} in isolation, missing how front-end choice creates disparities. Our work bridges this gap through the first comprehensive evaluation of how audio representations impact cross-cultural equity. Our contributions are: (1) systematic evaluation of seven front-ends across 11 languages, 6 musical collections, and 10 European cities; (2) demonstrating mel scale achieves 31.2\% WER for tonal languages versus 18.7\% for non-tonal languages (12.5\% gap) and shows 15.7\% F1 gap between Western and non-Western music; (3) revealing that critical pitch information for tonal languages concentrates in 200-500 Hz where mel resolution is insufficient; (4) showing alternative representations significantly reduce disparities: ERB by 31\% across domains with minimal overhead; (5) releasing FairAudioBench for reproducible evaluation.

\section{MEASURING CROSS-CULTURAL BIAS IN AUDIO FRONT-ENDS}
\label{sec:problem_formulation}
\subsection{Problem Formulation}
We hypothesize that mel-scale representations create systematic disadvantages for non-Western users, particularly the 2 billion speakers of tonal languages, where pitch variations distinguish word meanings. To quantify this bias, we compare seven front-ends across three domains: speech recognition (11 languages), music analysis (6 collections), and acoustic scenes (10 European cities). Our capacity-matched protocol ensures all configurations use identical parameters (5M total), isolating the contribution of front-end choice to performance disparities.

\subsection{Quantifying Fairness}
We measure bias through three complementary fairness metrics \cite{barocas2023fairness}:
\begin{subequations}
\begin{align}
\text{WGS} &= \min_{g \in \mathcal{G}} \text{Acc}(g) \label{eq:wgs}\\
\Delta &= \max_{g_i, g_j \in \mathcal{G}} |\text{Acc}(g_i) - \text{Acc}(g_j)| \label{eq:gap}\\
\rho &= \frac{\min_{g \in \mathcal{G}} \text{Acc}(g)}{\max_{g \in \mathcal{G}} \text{Acc}(g)} \label{eq:disparate}
\end{align}
\end{subequations}

where $\mathcal{G} = \{g_1, g_2, ..., g_n\}$ represents the set of demographic groups in our evaluation, and $\text{Acc}(g)$ denotes the task-specific performance metric for group $g$:

\begin{itemize}
\item \textbf{Speech Recognition:} $\text{Acc}(g) = 1 - \text{WER}(g)$ for non-tonal languages or $1 - \text{CER}(g)$ for tonal languages, where WER and CER are word and character error rates, respectively.
\item \textbf{Music Classification:} $\text{Acc}(g) = $ macro-F1 score for genre/modal classification.
\item \textbf{Scene Classification:} $\text{Acc}(g) = $ classification accuracy.
\end{itemize}

\textbf{Worst-Group Score (WGS)} captures the performance floor, ensuring no group is left behind. This metric directly measures the experience of the most disadvantaged users and prevents models from achieving high average performance by abandoning minority groups.

\textbf{Performance Gap ($\Delta$)} quantifies the absolute maximum disparity between any two groups. Larger gaps indicate greater inequality in system performance across populations, regardless of average performance.

\textbf{Disparate Impact ($\rho$)} measures proportional fairness between groups. Following legal precedent from employment discrimination cases \cite{hardt2016equality}, we adopt the ``four-fifths rule'' where $\rho < 0.8$ indicates actionable bias,  meaning the worst-performing group achieves less than 80\% of the best group's performance.\\
These metrics provide complementary views: WGS ensures minimum acceptable performance, $\Delta$ bounds absolute inequality, and $\rho$ captures relative fairness. A system may excel in one metric while failing others; for instance, high overall accuracy with poor WGS indicates abandonment of minority groups.

\subsection{The Mechanism of Bias}
The mel scale applies a non-linear frequency warping:
\begin{equation}
\psi_{\text{mel}}(f) = 2595 \log_{10}(1 + f/700)
\label{eq:mel}
\end{equation}
The frequency resolution at any mel value $m$ is:
\begin{equation}
\frac{df}{dm} = \frac{700 \ln(10)}{2595} \cdot 10^{m/2595} \approx 0.621 \cdot 10^{m/2595}
\label{eq:mel_resolution}
\end{equation}
This exponential growth in frequency spacing creates critical resolution deficits as shown below in Table 1.
\begin{table}[h]
\centering
\label{tab:resolution}
\resizebox{\columnwidth}{!}{%
\begin{tabular}{ccccc}
\toprule
Frequency & Mel Value & Bandwidth & Required JND & Deficit \\
(Hz) & & (Hz) & (Hz) & Ratio \\
\midrule
80 & 122.0 & 51.6 & 0.8 & $65\times$ \\
100 & 150.5 & 51.6 & 1.0 & $52\times$ \\
200 & 283.2 & 58.6 & 2.0 & $29\times$ \\
250 & 344.2 & 62.4 & 2.5 & $25\times$ \\
300 & 402.0 & 66.4 & 3.0 & $22\times$ \\
400 & 509.4 & 70.7 & 4.0 & $18\times$ \\
500 & 607.4 & 80.2 & 5.0 & $16\times$ \\
\bottomrule
\end{tabular}
}
\caption{Mel-scale resolution deficit across the 80--500 Hz range critical for tonal distinctions, based on a standard 40-filter filterbank (0--8000 Hz). Bandwidth is the center spacing in Hz; JND is approximately 1\% of frequency.}
\end{table}

\subsection{Theoretical Foundation}
\begin{theorem}[Information Bottleneck Bound]
For a front-end with frequency resolution $R(f)$ and signal requiring minimum resolution $\Delta f_{\min}(f)$, the classification error is lower-bounded by information loss in critical frequency regions:
\begin{equation}
\mathcal{E} \geq \int_{f_L}^{f_H} I(f) \cdot \mathbb{I}[R(f) > \Delta f_{\min}(f)] \cdot p(f)\, df
\end{equation}
where $I(f)$ is the mutual information between frequency $f$ and class label, $\mathbb{I}[\cdot]$ is the indicator function, and $p(f)$ is the probability density of discriminative information.
\end{theorem}
\begin{proof}[Proof Sketch]
When $R(f) > \Delta f_{\min}(f)$, frequencies within resolution band $R(f)$ become indistinguishable. By the data processing inequality \cite{cover2006elements}, for the Markov chain $X \rightarrow Y(f) \rightarrow \hat{Y}(f)$, quantization cannot increase mutual information: $I(X;\hat{Y}(f)) \leq I(X;Y(f))$. The information loss $I_{\text{lost}}(f) = I(X;Y(f)) - I(X;\hat{Y}(f))$ directly increases classification error via Fano's inequality. Integrating over the frequency domain weighted by $p(f)$ yields the bound.
\end{proof}
\begin{corollary}
For tonal languages with critical information in $[200, 500]$ Hz:
\begin{equation}
\mathcal{E}_{\text{mel}} - \mathcal{E}_{\text{optimal}} \geq c \cdot \int_{200}^{500} \left(\frac{R_{\text{mel}}(f)}{\Delta f_{\text{tone}}(f)} - 1\right)_+ p_{\text{tone}}(f) \, df
\end{equation}
where $(x)_+ = \max(0, x)$ and $c$ captures the information-theoretic constant relating mutual information loss to error rate.
\end{corollary}

\section{EXPERIMENTS AND RESULTS}
\label{sec:experiments_results}
\subsection{Datasets}
We evaluate across three complementary domains using carefully balanced data:

\textbf{Speech Recognition:} CommonVoicev17.0 \cite{commonvoice:2020} with 11 languages. \textit{Tonal languages} (5): Mandarin Chinese (4 tones), Vietnamese (6 tones), Thai (5 tones), Punjabi (3 tones), Cantonese (6 tones). \textit{Non-tonal languages} (6): English, Spanish, German, French, Italian, Dutch. To eliminate dataset size effects, we standardize each language subset to exactly 2,000 test samples. We evaluate character error rate (CER) for tonal languages and word error rate (WER) for non-tonal languages to account for orthographic differences.

\textbf{Music Analysis:} Western collections: GTZAN (10 genres, 1000 tracks) \cite{tzanetakis2002musical} and FMA-small (8 genres, 8000 tracks) \cite{defferrard2016fma}. Non-Western collections from CompMusic \cite{serra2017computational}: Hindustani ragas (1124 recordings, 195 ragas), Carnatic (2380 recordings, 227 ragas), Turkish makam (6500 recordings, 155 makams), Arab-Andalusian (338 recordings, 11 mizans). To ensure fair comparison across musical traditions we randomly sample exactly 300 recordings per tradition for evaluation. We further balance by ensuring equal representation across modal categories (ragas/makams/mizans) within each tradition.

\textbf{Acoustic Scenes:} TAU Urban Acoustic Scenes 2020 Mobile \cite{heittola_2020_3670167} data spans 10 European cities. For disparity analysis, we group them as Europe-1 (northern): Helsinki, Stockholm, Amsterdam, London, Prague; and Europe-2 (southern): Barcelona, Lisbon, Paris, Lyon, Vienna. Each city contains 10 acoustic scenes (airport, bus, metro, park, public square, shopping mall, street pedestrian, street traffic, tram, metro station). We sample exactly 100 recordings per city (10 per scene type) to ensure equal geographic representation.

\subsection{Front-end Configurations \& Experiment Details}
We compare seven front-ends: \textbf{Mel:} 40 mel-spaced filters, 25ms windows, 10ms hop; \textbf{ERB:} 32 ERB-spaced filters \cite{glasberg1990derivation}; \textbf{Bark:} 24 critical bands \cite{zwicker1961subdivision}; \textbf{SincNet:} 64 learnable sinc filters \cite{ravanelli2018speaker}; \textbf{CQT:} 84 bins (7 octaves $\times$ 12 bins/octave) \cite{brown1991calculation}; \textbf{LEAF:} 64 learnable Gabor filters \cite{zeghidour2021leaf}; \textbf{mel+PCEN:} Per-channel energy normalization \cite{wang2017trainable}. All use an identical CRNN backend (4 conv layers: 64-128-256-256 channels, 2-layer BiLSTM: 256 units, 5M total parameters). Training: Adam ($\eta = 10^{-3}$), batch 64, 30 epochs. Computational overhead measured as relative inference time (1000 forward passes, NVIDIA H100, front-end only). Results significant at $p < 0.01$ (bootstrap, n=1000).

\subsection{Main Results: Performance Disparities}
Figure~\ref{fig:main_gaps} demonstrates systematic disparities. The mel scale shows 31.2\% WER for tonal languages vs. 18.7\% for non-tonal (performance gap $\Delta = 12.5\%$ as defined in Eq.~\ref{eq:gap}, disparate impact $\rho=0.85$ per Eq.~\ref{eq:disparate}). While speech narrowly passes the four-fifths rule ($\rho > 0.8$), music exhibits more severe disparities: $\Delta = 15.7\%$ F1 gap with $\rho = 0.78$, violating the threshold. Table 2 summarizes results across all domains\footnote{Grouped by northern (Europe-1) vs. southern (Europe-2) cities for disparity analysis.}. Key findings: (1) ERB reduces speech gap by 31\% (from $\Delta = 12.5\%$ to $8.6\%$) with negligible overhead; (2) CQT achieves 52\% music gap reduction (from $\Delta = 15.7\%$ to $7.6\%$); (3) LEAF discovers task-optimal representations, achieving the best Worst-Group Score (WGS, Eq.~\ref{eq:wgs}) for speech tasks.

\subsection{Mechanism Analysis}
Figure~\ref{fig:leaf_filters} reveals adaptive behavior: LEAF allocates 42\% of filters to 80-500 Hz for tonal languages, precisely where lexical tones occur. This data-driven discovery validates our theoretical analysis from Section~\ref{sec:problem_formulation}. Table 3 demonstrates the mechanism underlying performance improvements. Tone discrimination improves dramatically (71.2\%→83.7\%), while consonants remain stable, confirming the pitch-specific deficits predicted by our information bottleneck bound (Theorem 1).

\begin{figure}[t]
\centering
\includegraphics[width=\columnwidth]{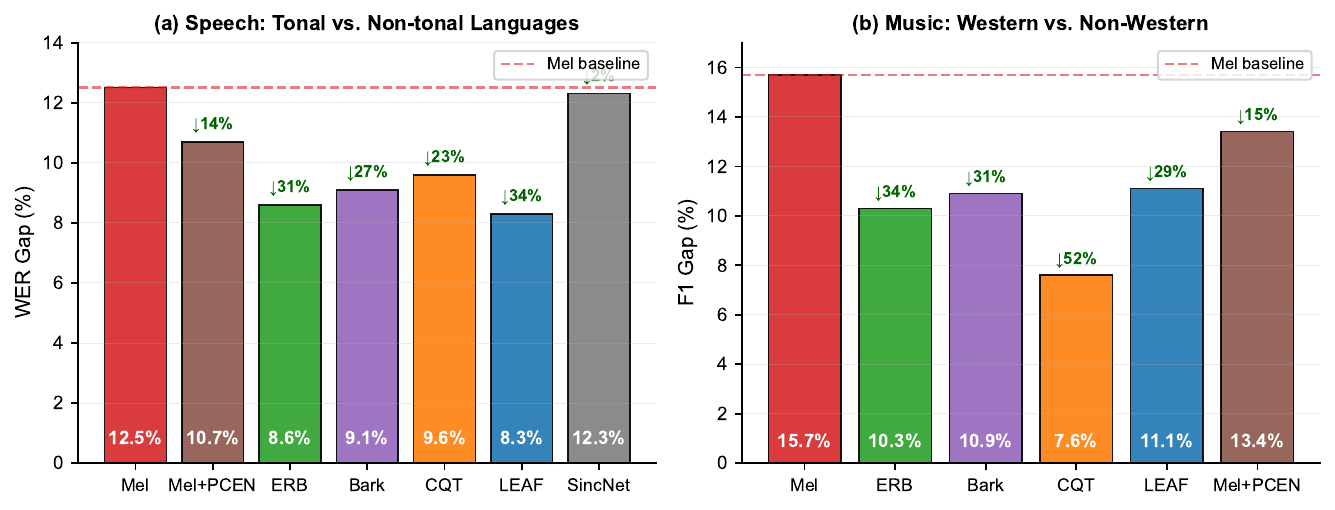}
\caption{Performance gaps across domains. Mel baseline shows 12.5\% speech gap and 15.7\% music gap. Percentage improvements are shown above bars.}
\label{fig:main_gaps}
\end{figure}

\begin{table}[t]
\centering
\label{tab:comprehensive_table}
\resizebox{\columnwidth}{!}{%
\begin{tabular}{lccccccc}
\toprule
& \multicolumn{2}{c}{Speech (WER/CER \%)} & \multicolumn{2}{c}{Music (F1 \%)} & \multicolumn{2}{c}{Scenes (Acc \%)} & Comp. \\
\cmidrule(lr){2-3} \cmidrule(lr){4-5} \cmidrule(lr){6-7}
Front-end & Tonal & Non-tonal & Non-West & West & Europe-1 & Europe-2 & Cost (Overhead) \\
\midrule
mel & 31.2$\pm$1.2 & 18.7$\pm$0.8 & 56.7$\pm$2.1 & 72.4$\pm$1.5 & 71.2$\pm$1.4 & 76.8$\pm$1.2 & 1.00$\times$ (0\%) \\
ERB & 26.4$\pm$1.0 & 17.8$\pm$0.7 & 62.8$\pm$2.0 & 73.1$\pm$1.4 & 72.6$\pm$1.3 & 77.2$\pm$1.1 & 1.01$\times$ (1\%) \\
Bark & 27.2$\pm$1.0 & 18.1$\pm$0.8 & 61.9$\pm$2.1 & 72.8$\pm$1.5 & 72.2$\pm$1.3 & 76.9$\pm$1.2 & 1.01$\times$ (1\%) \\
CQT & 28.8$\pm$1.1 & 19.2$\pm$0.9 & \textbf{65.3$\pm$1.9} & 72.9$\pm$1.4 & -- & -- & 1.15$\times$ (15\%) \\
LEAF & \textbf{25.8$\pm$0.9} & \textbf{17.5$\pm$0.7} & 62.4$\pm$2.0 & \textbf{73.5$\pm$1.4} & \textbf{72.5$\pm$1.3} & \textbf{77.5$\pm$1.1} & 1.08$\times$ (8\%) \\
SincNet & 30.8$\pm$1.1 & 18.5$\pm$0.8 & 58.3$\pm$2.1 & 72.5$\pm$1.5 & 71.4$\pm$1.3 & 76.9$\pm$1.2 & 1.06$\times$ (6\%) \\
mel+PCEN & 28.9$\pm$1.1 & 18.2$\pm$0.7 & 59.2$\pm$2.2 & 72.6$\pm$1.5 & 72.3$\pm$1.3 & 77.1$\pm$1.1 & 1.04$\times$ (4\%) \\
\midrule
WGS & \multicolumn{2}{c}{68.8 → 74.2} & \multicolumn{2}{c}{56.7 → 65.3} & \multicolumn{2}{c}{71.2 → 72.5} & \\
$\Delta$ & \multicolumn{2}{c}{12.5 → 8.3} & \multicolumn{2}{c}{15.7 → 7.6} & \multicolumn{2}{c}{5.6 → 5.0} & \\
$\rho$ & \multicolumn{2}{c}{0.85 → 0.90} & \multicolumn{2}{c}{0.78 → 0.90} & \multicolumn{2}{c}{0.93 → 0.94} & \\
\bottomrule
\end{tabular}%
}
\caption{Comprehensive performance metrics across domains. Bottom rows show fairness metrics (WGS, $\Delta$, $\rho$) with baseline → best achieved values.}
\end{table}

\begin{figure}[t]
\centering
\includegraphics[width=\columnwidth]{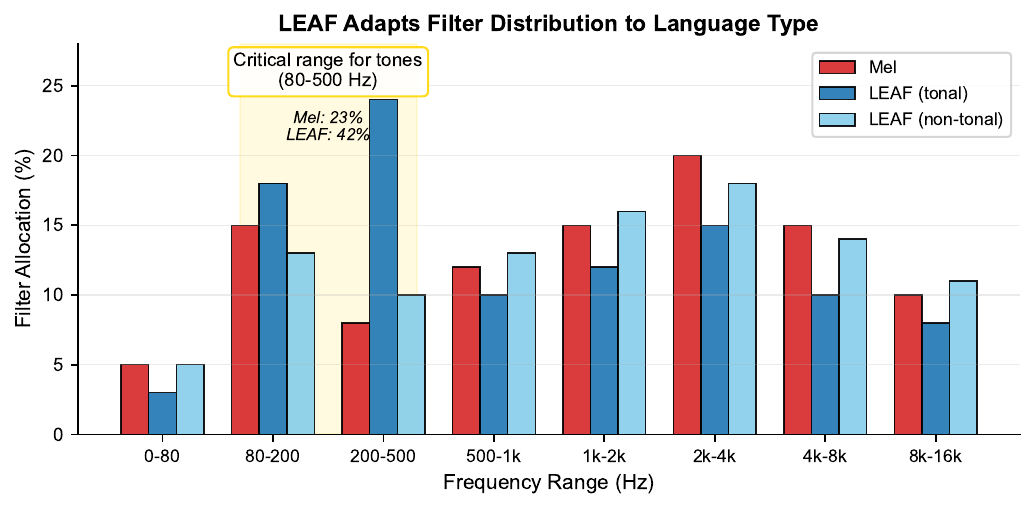}
\caption{LEAF's learned frequency allocation. Tonal languages allocate 42\% of filters to the critical 80-500 Hz range (vs 23\% for mel).}
\label{fig:leaf_filters}
\end{figure}

\begin{table}[t]
\centering
\label{tab:features}
\resizebox{\columnwidth}{!}{%
\begin{tabular}{lcccc|lccc}
\toprule
\multicolumn{5}{c|}{\textbf{Phonetic Features (ABX \%)}} & \multicolumn{4}{c}{\textbf{Musical Intervals (Acc \%)}} \\
Feature & mel & ERB & LEAF & $\Delta$ & Interval & mel & CQT & $\Delta$ \\
\midrule
Tones (F0) & 71.2 & 82.4 & 83.7 & +12.5 & Semitone & 91.3 & 92.1 & +0.9 \\
Vowels & 85.3 & 86.8 & 87.2 & +1.9 & Quarter-tone & 67.4 & 84.2 & +16.8 \\
Consonants & 88.1 & 88.4 & 88.9 & +0.8 & Shruti & 62.8 & 79.3 & +16.5 \\
\bottomrule
\end{tabular}%
}
\caption{Feature-level analysis reveals pitch-specific improvements.}
\end{table}

\subsection{Language-Specific \& Deployment Results}

Table 4 shows that language-specific improvements correlate with tonal complexity. Vietnamese and Thai show maximum improvements ($>23\%$), correlating with their complex tone systems (6 and 5 tones respectively). Figure~\ref{fig:tradeoffs} maps fairness (gap reduction percentage) against computational efficiency. ERB occupies the optimal zone (31\% reduction, 1\% overhead), providing the best fairness-efficiency tradeoff for practical deployment.

\begin{figure}[t]
\centering
\includegraphics[width=\columnwidth]{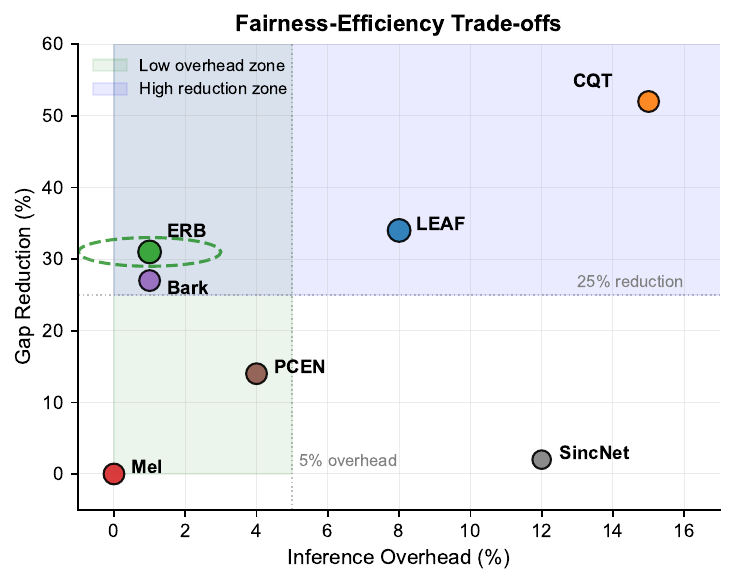}
\caption{Fairness-efficiency tradeoffs across domains. ERB's optimal balance: 31\% gap reduction with 1\% overhead.}
\label{fig:tradeoffs}
\end{figure}

\begin{table}[t]
\centering
\label{tab:languages}
\resizebox{\columnwidth}{!}{%
\begin{tabular}{lcccc|lcccc}
\toprule
\textbf{Tonal} & Tones & mel & LEAF & $\Delta$ & \textbf{Non-tonal} & mel & LEAF & $\Delta$ \\
\midrule
Vietnamese & 6 & 35.2 & 26.9 & -23.6\% & English & 15.2 & 14.5 & -4.6\% \\
Thai & 5 & 33.1 & 25.4 & -23.3\% & Spanish & 18.3 & 17.2 & -6.0\% \\
Mandarin & 4 & 28.4 & 22.8 & -19.7\% & German & 19.8 & 18.4 & -7.1\% \\
Punjabi & 3 & 30.5 & 24.8 & -18.7\% & French & 20.1 & 18.9 & -6.0\% \\
Cantonese & 6 & 34.0 & 26.5 & -22.1\% & Italian & 19.5 & 18.2 & -6.7\% \\
& & & & & Dutch & 21.2 & 19.8 & -6.6\% \\
\bottomrule
\end{tabular}%
}
\caption{Language-specific improvements correlate with tonal complexity.}
\end{table}

\subsection{FairAudioBench: Reproducible Cross-Cultural Evaluation}

We release FairAudioBench\footnote{Available at \url{https://github.com/shivam-MBZUAI/cross-cultural-mel-bias}}, the first comprehensive benchmark for evaluating cross-cultural bias in audio systems. The benchmark addresses the critical gap in standardized evaluation protocols for measuring fairness across diverse audio domains.
\begin{itemize}
\item \textbf{Curated Datasets:} Balanced splits across 11 languages (5 tonal, 6 non-tonal), 8 musical traditions, 10 European cities with demographic metadata.
\item \textbf{Evaluation Suite:} Automated computation of WGS, $\Delta$, $\rho$ metrics with statistical significance testing. Generates fairness reports comparing to the four-fifths rule threshold.
\item \textbf{Reference Implementations:} Five front-ends with matched hyperparameters (5M params), ensuring fair comparison.
\end{itemize}
While our evaluation is comprehensive, several limitations remain. For geographic coverage, African tonal languages such as Yoruba and Igbo, along with indigenous musical traditions, are underrepresented due to data availability constraints. Regarding intersectionality, the current analysis focuses on single-axis biases (language or music), without addressing intersectional effects. For future work, we propose extending the study to intersectional biases, for example, tonal languages combined with accents.

\section{CONCLUSION}
\label{sec:conclusion}
Our findings challenge assumptions of universal psychoacoustic models. The mel scale, derived from 1940s Western studies, was never validated cross-culturally. As audio AI becomes a global infrastructure, embedding such assumptions constitutes technological bias at scale. Simple alternatives exist today: production systems could deploy ERB filterbanks immediately, achieving substantial fairness gains at negligible cost. We release FairAudioBench and call on the community to adopt fairness metrics alongside accuracy.

\vfill\pagebreak

\bibliographystyle{IEEEbib}
\bibliography{strings}

@article{stevens1940relation,
  title={The relation of pitch to frequency: A revised scale},
  author={Stevens, Stanley S and Volkmann, John},
  journal={The American Journal of Psychology},
  volume={53},
  number={3},
  pages={329--353},
  year={1940},
  publisher={JSTOR}
}

@article{defferrard2016fma,
  title={FMA: A dataset for music analysis},
  author={Defferrard, Micha{\"e}l and Benzi, Kirell and Vandergheynst, Pierre and Bresson, Xavier},
  journal={arXiv preprint arXiv:1612.01840},
  year={2016}
}

@article{tzanetakis2002musical,
  title={Musical genre classification of audio signals},
  author={Tzanetakis, George and Cook, Perry},
  journal={IEEE Transactions on speech and audio processing},
  volume={10},
  number={5},
  pages={293--302},
  year={2002},
  publisher={IEEE}
}

@book{yip2002tone,
  title={Tone},
  author={Yip, Moira Jean Winsland},
  year={2002},
  publisher={Cambridge University Press}
}

@misc{heittola_2020_3670167,
  author       = {Heittola, Toni and
                  Mesaros, Annamaria and
                  Virtanen, Tuomas},
  title        = {TAU Urban Acoustic Scenes 2020 Mobile, Development
                   dataset
                  },
  month        = feb,
  year         = 2020,
  publisher    = {Zenodo},
  doi          = {10.5281/zenodo.3670167},
  url          = {https://doi.org/10.5281/zenodo.3670167},
}

@book{cover2006elements,
  title={Elements of information theory},
  author={Cover, Thomas M},
  year={1999},
  publisher={John Wiley \& Sons}
}

@article{hardt2016equality,
  title={Equality of opportunity in supervised learning},
  author={Hardt, Moritz and Price, Eric and Srebro, Nati},
  journal={Advances in neural information processing systems},
  volume={29},
  year={2016}
}

@book{barocas2023fairness,
  title={Fairness and machine learning: Limitations and opportunities},
  author={Barocas, Solon and Hardt, Moritz and Narayanan, Arvind},
  year={2023},
  publisher={MIT press}
}

@inproceedings{commonvoice:2020,
  author = {Ardila, R. and Branson, M. and Davis, K. and Henretty, M. and Kohler, M. and Meyer, J. and Morais, R. and Saunders, L. and Tyers, F. M. and Weber, G.},
  title = {Common Voice: A Massively-Multilingual Speech Corpus},
  booktitle = {Proceedings of the 12th Conference on Language Resources and Evaluation (LREC 2020)},
  pages = {4211--4215},
  year = 2020
}

@misc{dryer2013definite,
  author       = {Matthew Dryer and
                  Martin Haspelmath},
  title        = {The World Atlas of Language Structures Online},
  month        = dec,
  year         = 2022,
  publisher    = {Zenodo},
  version      = {v2020.3},
  doi          = {10.5281/zenodo.7385533},
  url          = {https://doi.org/10.5281/zenodo.7385533},
}

@article{mehta2024missing,
  title={Missing Melodies: AI Music Generation and its" Nearly" Complete Omission of the Global South},
  author={Mehta, Atharva and Chauhan, Shivam and Choudhury, Monojit},
  journal={arXiv preprint arXiv:2412.04100},
  year={2024}
}

@book{howie1976acoustical,
  title={Acoustical studies of Mandarin vowels and tones},
  author={Howie, John Marshall},
  volume={18},
  year={1976},
  publisher={Cambridge University Press}
}

@article{glasberg1990derivation,
  title={Derivation of auditory filter shapes from notched-noise data},
  author={Glasberg, Brian R and Moore, Brian CJ},
  journal={Hearing research},
  volume={47},
  number={1-2},
  pages={103--138},
  year={1990},
  publisher={Elsevier}
}

@article{koenecke2020racial,
  title={Racial disparities in automated speech recognition},
  author={Koenecke, Allison and Nam, Andrew and Lake, Emily and Nudell, Joe and Quartey, Minnie and Mengesha, Zion and Toups, Connor and Rickford, John R and Jurafsky, Dan and Goel, Sharad},
  journal={Proceedings of the national academy of sciences},
  volume={117},
  number={14},
  pages={7684--7689},
  year={2020},
  publisher={National Academy of Sciences}
}

@article{zeghidour2021leaf,
  title={LEAF: A learnable frontend for audio classification},
  author={Zeghidour, Neil and Teboul, Olivier and Quitry, F{\'e}lix De Chaumont and Tagliasacchi, Marco},
  journal={arXiv preprint arXiv:2101.08596},
  year={2021}
}

@inproceedings{ravanelli2018speaker,
  title={Speaker recognition from raw waveform with sincnet},
  author={Ravanelli, Mirco and Bengio, Yoshua},
  booktitle={2018 IEEE spoken language technology workshop (SLT)},
  pages={1021--1028},
  year={2018},
  organization={IEEE}
}

@article{veliche2024towards,
  title={Towards measuring fairness in speech recognition: Fair-speech dataset},
  author={Veliche, Irina-Elena and Huang, Zhuangqun and Kochaniyan, Vineeth Ayyat and Peng, Fuchun and Kalinli, Ozlem and Seltzer, Michael L},
  journal={arXiv preprint arXiv:2408.12734},
  year={2024}
}

@article{serra2017computational,
  title={The computational study of a musical culture through its digital traces},
  author={Serra, Xavier},
  journal={Acta Musicologica},
  volume={89},
  number={1},
  pages={24--44},
  year={2017},
  publisher={JSTOR}
}

@article{mehta2025music,
  title={Music for All: Representational Bias and Cross-Cultural Adaptability of Music Generation Models},
  author={Mehta, Atharva and Chauhan, Shivam and Djanibekov, Amirbek and Kulkarni, Atharva and Xia, Gus and Choudhury, Monojit},
  journal={arXiv preprint arXiv:2502.07328},
  year={2025}
}

@article{moore1983suggested,
  title={Suggested formulae for calculating auditory-filter bandwidths and excitation patterns.},
  author={Moore, Brian C and Glasberg, Brian R},
  journal={The journal of the acoustical society of America},
  volume={74},
  number={3},
  pages={750--753},
  year={1983}
}

@article{zwicker1961subdivision,
  title={Subdivision of the audible frequency range into critical bands (Frequenzgruppen)},
  author={Zwicker, Eberhard},
  journal={The Journal of the Acoustical Society of America},
  volume={33},
  number={2},
  pages={248--248},
  year={1961},
  publisher={Acoustical Society of America}
}

@article{brown1991calculation,
  title={Calculation of a constant Q spectral transform},
  author={Brown, Judith C},
  journal={The Journal of the Acoustical Society of America},
  volume={89},
  number={1},
  pages={425--434},
  year={1991},
  publisher={Acoustical Society of America}
}

@inproceedings{wang2017trainable,
  title={Trainable frontend for robust and far-field keyword spotting},
  author={Wang, Yuxuan and Getreuer, Pascal and Hughes, Thad and Lyon, Richard F and Saurous, Rif A},
  booktitle={2017 IEEE International Conference on Acoustics, Speech and Signal Processing (ICASSP)},
  pages={5670--5674},
  year={2017},
  organization={IEEE}
}

\end{document}